\documentclass[12pt]{article}
\usepackage{subfigure,a4}
\usepackage{multirow}
\usepackage{epsfig,ulem}
\usepackage{dcolumn}
\usepackage{bm}

\newcommand{\be}{\begin{eqnarray}}
\newcommand{\ee}{\end{eqnarray}}

\usepackage[english]{babel}
\usepackage{amssymb}
\usepackage{amsmath}
\usepackage{amsfonts}
\usepackage{amsbsy}
\usepackage{indentfirst}
\usepackage{graphicx}
\usepackage{color}

\newcommand{\C}{\mathbb{C}}
\newcommand{\R}{\mathbb{R}}
\newcommand{\Z}{\mathbb{Z}}

\newcommand{\e}{\mathrm{e}}

\newtheorem{theorem}{Theorem}[section]
\newtheorem{pro}{Proposition}[section]
\newtheorem{rmk}{Remark}[section]
\newtheorem{corol}{Corollary}[section]

\newtheorem{Exa}{Example}[section]
\newenvironment{proof}[1][Proof]{\noindent\textbf{#1.} }{\ \rule{0.5em}{0.5em}}
\def\d{{\rm d}}

\def\a{\alpha}
\def\diag{{\rm diag}}
\def\<{\langle}
\def\>{\rangle}

\def\diag{{\rm diag}}
\def\Tr{{\rm Tr}}

\newcommand{\beq}{\begin{equation}}
\newcommand{\eeq}{\end{equation}}

\newcommand{\bmat}{\begin{displaymath}}
\newcommand{\emat}{\end{displaymath}}

\def\a{\alpha}
\def\b{\beta}
\def\c{\gamma}
\def\ta{\a_{11}}
\def\tb{\a_{22}}
\def\tc{\b_{12}}
\def\td{\a_{12}}
\def\te{\b_{11}}
\def\tf{\b_{22}}
\def\1{{\bf 1}}

\begin{document}

\title {A quantum system with a non-Hermitian Hamiltonian
}
\author{
N. Bebiano\footnote{ University of Coimbra, CMUC, Department of Mathematics,
P 3001-454 Coimbra, Portugal (bebiano@mat.uc.pt)},
J.~da Provid\^encia\footnote{University of Coimbra, CFisUC, Department of Physics,
P 3004-516 Coimbra, Portugal
(providencia@teor.fis.uc.pt)},
S.~Nishiyama\footnote{University of Coimbra, CFisUC, Department of Physics,
P 3004-516 Coimbra, Portugal
(seikoceu@khe.biglobe.ne.jp)}
~ and J.P. da
Provid\^encia\footnote{Univ. of Beira Interior, Department of Physics,
P-6201-001 Covilh\~a, Portugal
(joaodaprovidencia@daad-alumni.de)}
 }
\maketitle
\begin{abstract}
The relevance in Physics of non-Hermitian operators with real
eigenvalues is being widely recognized not only in quantum mechanics
but also in other areas, such as quantum optics, quantum fluid
dynamics and quantum field theory.
In this note,  a quantum system
described by a non-Hermitian Hamiltonian, which is constituted by two types of interacting bosons,
is investigated. The real eigenvalues of the Hamiltonian are explicitly
determined, as well as complete biorthogonal sets of eigenfunctions
of the Hamiltonian and its adjoint. The diagonal representation of $H$
is obtained using pseudo-bosonic operators.
\end{abstract}
\section{Introduction}
In Quantum Mechanics (QM), the states of a particle, or system of
particles, are represented by vectors in a Hilbert space $\cal H$,
endowed with an inner product $\langle\cdot,\cdot\rangle$.
Measurable physical quantities are represented
by Hermitian (in the von Neumann terminology) or self-adjoint operators
in $\cal H$,
called {\it observables}.
Hermiticity ensures, in particular, that the eigenvalues (energy levels) are real and the unitarity of time evolution.
The expectation value  of $H$ in the state represented by the state vector $\Psi$ is given by the
Rayleigh quotient
$$\frac{\langle H\Psi,\Psi\rangle}{\langle\Psi,\Psi\rangle},~~\Psi\in{\cal H}.$$
When $H$ is non-Hermitian we cannot interpret the Rayleigh quotient as the average value of measurements,
because it may become a complex number, which violates the measurement axiom of QM.
On the other hand, if we regard a non-Hermitian operator $H$
as the  generator of the system dynamics
via the Sch\"odinger equation,
$$i\frac{\d }{\d t}\Psi=H\Psi,$$
other problems arise. The norm of the state vector $\Psi$
is not preserved, as
\begin{eqnarray*}
&&\frac{\d \langle\Psi,\Psi\rangle}{\d t}=\langle\frac{\d\Psi}{\d t},\Psi\rangle+\langle\Psi,\frac{\d\Psi}{\d t}\rangle\\
&&=\langle-iH\Psi,\Psi\rangle+\langle\Psi,-iH\Psi\rangle\\
&&=\langle\Psi,i(H^*-H)\Psi\rangle\neq 0.
\end{eqnarray*}
Henceforth, the Hamiltonian does not  generate a unitary time evolution.
Summing up, the standard norm and the Rayleigh quotient
 with this inner product    are not adequate for the physical interpretation.

In the last decades, the possibility of  a non Hermitian operator to play
the  role of a Hamiltonian in the standard formalism of QM
challenged researchers. Moreover, the occurrence of concrete models involving non Hermitian Hamiltonians
 motivated a remarkable development of this area (see
\cite{bagarello,*,bebiano*,[1],[2],mostafa,providencia,scholtz,[3]} and their references).
In this work, we introduce a non Hermitian Hamiltonian describing two types of interacting bosons
and we show how to overcome the apparent conflict raised by non Hermiticity with
the standard formalism of QM.

This note is organized  as follows. In Section \ref{S2}, we introduce our model
represented by a non Hermition  Hamiltonian $H$.
In Section \ref{S3}, the equation of motion method (EMM) is  used to determine the eigenvalues
and eigenvectors of $H$, and the arising matrix eigenproblem is investigated.
In Section \ref{S4}, the diagonal form of $H$ is obtained in terms of pseudo-bosonic operators.
In Section \ref{S5}, the biorthogonal systems of eigenvectors of the Hamiltonians $H$ and
of its adjoint $H^*$ are presented.
In Section \ref{S6}, the so called metric operator
$\e^{-S}$, which renders the
Hamiltonian Hermitian, that is, $\e^{-S}H=H^*e^{-S}$, is defined.
In Section \ref{S7}, the
physical Hilbert space, which is  required for the
physical interpretation of the model, is introduced, by defining a
new inner product,
called the physical inner product,
in a certain subspace of $\cal H$.
In Section \ref{S8},
some considerations on statistical mechanics for the non-Hermitian setting are presented.

\section{Non Hermitian Hamiltonian $H$
describing two types of interacting bosons}\label{S2}

We shall consider the Hamiltonian $H$
describing a system of two types of interacting bosons $a_i,~i=1,2,$ defined as
\begin{eqnarray}&&H=\a_{11}a^*_1a_1+\a_{22}a^*_2a_2+\a_{12}(a^*_1a_2-a^*_2a_1)\nonumber\\&&
+{1\over2}\b_{11}(a_1^{*2}-a_1^2)+{1\over2}\b_{22}(a_2^{*2}-a_2^2)+\b_{12}(a_1^{*}a_2^*-a_2a_1)\quad
\a_{ij},\b_{ij}\in\R \label{Ham},\end{eqnarray}
and acting on an infinite
dimensional separable  Hilbert space $\cal H$.

The present model is a generalization of the previously considered models in \cite{bebiano****,bebiano*****},
whose Hamiltonians are expressed in terms of the bosonic generators of the $su(1,1)$ and $su(2)$ algebras.
Here, $H$ is expressed in terms of the operators $a_i^*a_i,~a^*_i a_j,$ $ a^*_ia^*_j,~a_i a_j,~1\leq i\leq j\leq 2,$ which
generate the $sp(2n)$ algebra for $n=2$, of which   $su(1,1)$ and $su(2)$ are sub-algebras.

As it is well known, the Hamiltonian $H$ in (\ref{Ham}), may be
uniquely written in the Cartesian decomposition $H={\rm Re}H+i{\rm
Im}H$, where ${\rm Re}H=(H+H^*)/2$ is the {\it Hermitian part} of $H$ and
${\rm Im}H=(H-H^*)/(2i)$ is the {\it imaginary part}.
Without loss of
generality, by a unitary similarity transformation, we may assume that 
${\rm Re}H$ has diagonal form, and so (\ref{Ham}) represents a general Hamiltonian
describing a system of two types of interacting bosons.
We assume that  $\a_{ii}>0,$ as is required by
physical significance. This condition ensures that the eigenvalues of
the Hermitian part of $H$,  are positive.
As usual,  $a_i^*$ denotes the
adjoint operator of $a_i$ with respect to the inner product in $\cal
H$. These operators act on a dense subspace $\cal D$ of $\cal H$
and  satisfy the so
called Weyl-Heisenberg commutation relations,
\begin{eqnarray}&&[a_i,a_i^*]={\bf1},~[a_i^*,a_i^*]=
[a_i,a_i]=0,~i=1,2,\nonumber\\&&\label{W-H}[a_i,a_j^*]=[a_i^*,a_j^*]=
[a_i,a_j]=0,~i\neq j=1,2,\end{eqnarray}
where $\bf 1$ is the identity
operator in $\cal H$.

Conventionally, the operators $a_i^*$ are called {\it creation
operators}, while the $a_i$ are the {\it annihilation operators}.
This means that, if $\Phi$ is an eigenvector of the number operator $N_i=a_i^*a_i$
associated with the eigenvalue $n_i$, which is a nonnegative integer,
$$N_i\Phi=n_i\Phi,$$ then
$a_i^*\Phi$ is an eigenvector associated with the eigenvalue
$n_i+1$,
$$N_i\Phi=n_i\Phi,~~N_ia_i^*\Phi=(n_i+1)a_i^*\Phi,$$
  while, for $n_i\neq0$, $a_i\Phi$ is an eigenvector
associated with the eigenvalue $n_i-1$,
$$~~N_ia_i\Phi=(n_i-1)a_i\Phi~\text{if}~n_i-1\geq0.$$ \color{black}

The algebra generated by
$a_1,~a_2,~a_1^*,~a_2^*$ and satisfying the above commutation
relations is called a Weyl-Heisenberg algebra. The vector $0\neq\Phi_0\in\cal D$, which  denotes
the {\it vacuum} of $a_1,a_2$, i.e., $a_i\Phi_0=0,~i=1,2,$
and the family of semi bounded operators
$${\cal F}_\Phi=\{\Phi_{n_1n_2}=a_1^{*n_1}a_2^{*n_2}\Phi_0,~n _1,n_2\in\Z^+\},$$
constitute a basis for the Hilbert space $\cal H$ \cite{dirac}.
A simple reminder: ${\cal F}_\Phi$ is a basis  of $\cal
H$ if any $v\in\cal H$ admits a unique decomposition in terms of the
elements $\Phi_{n_1n_2}$ of ${\cal F}_\Phi$.
The linear span  of {\color{red}${\cal F}_\Phi$}
is an infinite dimensional vector space,
{\color{red}which is dense in } $\cal H$.

As it will be shown, the spectrum of  ${\rm Re}H$ is real and
discrete, which means that it consists of real, simple eigenvalues,
$$(\a_{11}a_1^*a_1+\a_{22}a_2^*a_2)\Phi_{n_1n_2}=(n_1\a_{11}+n_2\a_{22})\Phi_{n_1n_2}.$$

\section{The EMM and matrix preliminaries}\label{S3}

In order to determine
the eigenvectors and eigenvalues of $H$
we consider operators of the
form
\begin{equation}\Theta=x_1a^*_1+x_2a^*_2+y_1a_1+y_2a_2,\quad x_1,x_2,y_1,y_2\in \C\label{Theta}\end{equation}
satisfying
$$[H,\Theta]=H\Theta-\Theta H=\lambda\Theta,\quad\lambda\in\R.$$
Having in mind (\ref{Ham}) and (\ref{W-H}) in (\ref{Theta}), this leads to the
matrix eigenvalue-eigenvector problem
$$A(x_1,x_2,y_1,y_2)^T=\lambda(x_1,x_2,y_1,y_2)^T,\quad(x_1,x_2,y_1,y_2)^T\in\C^4\backslash\{0\},$$
where
$$A=\left[\begin{matrix}\a_{11}&\a_{12}&-\b_{11}&-\b_{12}\\
-\a_{12}&\a_{22}&-\b_{12}&-\b_{22}\\
-\b_{11}&-\b_{12}&-\a_{11}&\a_{12}\\
-\b_{12}&-\b_{22}&-\a_{12}&-\a_{22}\end{matrix}\right].$$
Notice
that, for $\alpha_{12}=0,$ the matrix $A$ is Hermitian and so it has real eigenvalues.
For $\beta_{11}=\beta_{22}=0$, $A$ is $J$-symmetric for
$J=\diag(1,-1,-1,1)$, that is, $JAJ=A^T$. Otherwise, this matrix is
neither symmetric nor $J$-symmetric. However, $\Omega A$ is symmetric for
$$\Omega=\left[\begin{matrix}0&0&-1&0\\0&0&0&-1\\1&0&0&0\\
0&1&0&0\end{matrix}\right].$$
and its eigenvalues are real or occur in conjugate pairs.
We observe that $\Omega^2=-I_4$ and that
$\Omega^T=-\Omega.$ Thus, $\Omega A=A^T\Omega^T$ and, consequently,
$-A^T=\Omega A\Omega^T$.

Assume for the moment that the eigenvalues of $A$ are real. If
$u$ is an eigenvector of $A^T$
associated with the eigenvalue $\lambda$, then $v=\Omega^Tu=-\Omega u$ is an
eigenvector of $A$, associated with the eigenvalue $-\lambda$.
Further, $\lambda$ being an eigenvalue of $A^T$, is also an eigenvalue of $A$.
So, the eigenvalues of $A$ occur in symmetric pairs.

\begin{pro}
If $|a_{12}|$ is sufficiently small,  the eigenvalues of $A$ are real and simple, and occur in symmetric pairs.
\end{pro}

\begin{proof}
\color{black}
The characteristic polynomial of $A$ may be written as
$$\det(A-\lambda I_4)=C-B\lambda^2+\lambda^4,$$where
 \begin{eqnarray*}&&B=\ta^2 + \tb^2  +2 \tc^2 - 2 \td^2 + \te^2
+ \tf^2\\&& C=\ta^2 \tb^2 + 2 \ta \tb \tc^2 + \tc^4 + 2 \ta \tb
\td^2 - 2 \tc^2 \td^2 +
    \td^4 + \tb^2 \te^2\\&& - 2 \tc^2 \te \tf + 2 \td^2 \te \tf + \ta^2 \tf^2 +
    \te^2 \tf^2.
\end{eqnarray*}
From the characteristic equation $\det(A-\lambda I_4)=0,$ we obtain
$$\lambda=\pm\sqrt{{1\over2}B\pm{1\over2}\sqrt{B^2-4C}}.$$
If $\a_{12}= 0$, then $B>0$ and
$B^2-4C$ becomes
$$(\ta^2 - \tb^2 + \te^2 - \tf^2)^2 +
4\tc^2 ((\ta - \tb)^2 + (\te + \tf)^2)>0,$$ so that the eigenvalues
$\lambda$ are real, as $B^2>B^2-4C$. The eigenvalues of $A$ are still real if
$|\a_{12}|$ is small enough, due to the continuity of the
eigenvalues on the entries of $A$. Hence, if $|\a_{12}|$ is
sufficiently small, the matrix $A$ has two positive eigenvalues and
two negative eigenvalues.
\end{proof}

\begin{pro}Assume   the eigenvalues of $A$ ordered as $\lambda_1>\lambda_2>0>\lambda_4>\lambda_3$, and satisfying
$\lambda_3=-\lambda_1,~\lambda_4=-\lambda_2.$ The respective eigenvectors $v_1,v_2,v_4,v_3$ may be normalized as
$v_1^T\Omega v_3=v_2^T\Omega v_4=1.$
Moreover, $v_i^T\Omega v_j=0,$ for $\{(i,j):i,j=1,2,3,4\}\backslash\{(1,3),(2,4),(3,1),(4,2)\}.$
\label{P2.2}\end{pro}
\color{black}

\begin{proof}
\color{red}Let $v,~v'$ be eigenvectors of $A$ associated,
respectively, with the eigenvalue $\lambda$ and
$\lambda'$. If $\lambda=\lambda'$,
it follows that $v'^T\Omega v=0$, so we may have $v'^T\Omega v\neq0,$ only if  $\lambda'\neq\lambda$
and $\lambda'\lambda<0.$
\color{red}We assume that the eigenvectors are normalized so that $v_1^T\Omega v_3=~v_2^T\Omega v_4=1$.
Thus $~v_3^T\Omega v_1=~v_4^T\Omega v_2=-1$. Otherwise $v_i^T\Omega v_j=0.$
\end{proof}\\

Throughout, we will denote the entries of $v_i$ as follows:
\begin{equation}
v_i=(v_{i,1},v_{i,2},v_{i,3},v_{i,4})^T=(x_{1}^{(i)},x_{2}^{(i)},y_{1}^{(i)},y_{2}^{(i)})^T\in\C^4\backslash
\{0\},~~i=1,2,3,4.\label{vi}
\end{equation}
\begin{pro}
The matrix whose rows are the (transposed) eigenvectors of $A$
$$U=\left[\begin{matrix}x_1^{(1)}&x_2^{(1)}&y_1^{(1)}&y_2^{(1)}\\
x_1^{(2)}&x_2^{(2)}&y_1^{(2)}&y_2^{(2)}\\
x_1^{(3)}&x_2^{(3)}&y_1^{(3)}&y_2^{(3)}\\
x_1^{(4)}&x_2^{(4)}&y_1^{(4)}&y_2^{(4)}\end{matrix}\right]$$ is
invertible and $U^{-1}=\Omega U^T\Omega^T.$
\end{pro}

\begin{proof}
 By Proposition \ref{P2.2}, the vectors
$(x_1^{(i)},x_2^{(i)},y_1^{(i)},y_2^{(i)})^T$ may be normalized so that
$$U\Omega U^T=\Omega,$$which implies
\begin{equation}\label{ZZT}U\Omega U^T\Omega^T=I_4.\end{equation} As a consequence, $(U^{T})^{-1}=\Omega U\Omega^T$ and
$U^{-1}=\Omega U^T\Omega^T$.
\end{proof}\\

\begin{rmk}
If $U\Omega U^T\Omega^T=V\Omega V^T\Omega^T=I_4$ it follows that
\begin{eqnarray*}&&(UV)\Omega(UV)^T\Omega^T=U V \Omega V^T\Omega^T\Omega U^T\Omega^T=I_4.
\end{eqnarray*}
 Matrices satisfying (\ref{ZZT}) constitute the sympletic group $Sp(2n)$, for $n=2$. 
\end{rmk}

\section{$\cal D$-Dynamical pseudo-bosons and the diagonal form of $H$}\label{S4}
Let us define
$$\widehat\Theta_i=v_{i,1}a^*_1+v_{i,2}a^*_2+v_{i,3}a_1+v_{i,4}a_2,~~i=1,2,3,4,$$
with $v_i$ as in (\ref{vi}).
 {\color{red}Recalling  that
$\lambda_1>\lambda_2>0>\lambda_4>\lambda_3$,}  the operators
\begin{equation}\Theta^\ddag_i=\widehat\Theta_i=v_{i,1}a^*_1+v_{i,2}a^*_2+v_{i,3}a_1+v_{i,4}a_2,~~i=1,2,\label{Thetaddag}
\end{equation}
are creation operators because, when applied to an eigenvector of
$H$, since they lead to a new eigenvector with the eigenvalue increased by
$\lambda_i$. Further,
\begin{equation}\Theta_{i}=\widehat\Theta_{i+2}=v_{i+2,1}a^*_1+v_{i+2,2}a^*_2+v_{i+2,3}a_1+v_{i+2,4}a_2,~~i=1,2\label{Theta}
\end{equation}
are annihilation operators, as, when they are applied to an
eigenvector of $H$, they give rise to a new eigenvector of $H$
associated with a new eigenvalue decreased by $\lambda_{i}.$

We may write (\ref{Thetaddag}) and (\ref{Theta}) in matrix form
$$\left[\begin{matrix}\Theta_1^\ddag\\\Theta_2^\ddag\\\Theta_1\\\Theta_2\end{matrix}\right]
=\left[\begin{matrix}x_1^{(1)}&x_2^{(1)}&y_1^{(1)}&y_2^{(1)}\\
x_1^{(2)}&x_2^{(2)}&y_1^{(2)}&y_2^{(2)}\\
x_1^{(3)}&x_2^{(3)}&y_1^{(3)}&y_2^{(3)}\\
x_1^{(4)}&x_2^{(4)}&y_1^{(4)}&y_2^{(4)}\end{matrix}\right]\left[\begin{matrix}a_1^*\\a_2^*\\a_1\\a_2\end{matrix}\right],
$$ and so$$
\left[\begin{matrix}a_1^*\\a_2^*\\a_1\\a_2\end{matrix}\right]=\left[\begin{matrix}x_1^{(1)}&x_2^{(1)}&y_1^{(1)}&y_2^{(1)}\\
x_1^{(2)}&x_2^{(2)}&y_1^{(2)}&y_2^{(2)}\\
x_1^{(3)}&x_2^{(3)}&y_1^{(3)}&y_2^{(3)}\\
x_1^{(4)}&x_2^{(4)}&y_1^{(4)}&y_2^{(4)}\end{matrix}\right]^{-1}
\left[\begin{matrix}\Theta_1^\ddag\\\Theta_2^\ddag\\\Theta_1\\\Theta_2\end{matrix}\right].$$
The
operators $\Theta^\ddag_i,\Theta_i$ satisfy the Weyl-Heisenberg commutation relations
$$[\Theta_1, \Theta_1^\ddag]=[\Theta_2, \Theta_2^\ddag]=1,~~[\Theta_1, \Theta_2^\ddag]=[\Theta_2, \Theta_1^\ddag]
=[\Theta_1^\ddag, \Theta_2^\ddag]=[\Theta_1, \Theta_2]=0,$$
as may be easily verified.
For instance
$$[\Theta_2,\Theta^\ddag_1]=y_1^{(4)}x_1^{(1)}+y_2^{(4)}x_2^{(1)}-y_1^{(4)}x_1^{(1)}-y_2^{(4)}x_2^{(1)}=0.$$
Since
the Weyl-Heisenberg commutation relations are satisfied,
but $\Theta^\ddag_i\neq\Theta_i$, and these operators have a vacuum state $\Psi_0,$
we say that they describe {\it pseudo-bosons}
\cite{bagarello}.

\begin{theorem}\label{T4.1}
For $\Theta_i$ and $\Theta_i^\ddag$ as in (\ref{Theta}) and (\ref{Thetaddag}), $\lambda_1$ and $\lambda_2$ as in Proposition 3.2,
the Hamiltonian $H$ may be written in the diagonal form,
$$H=E_0+\lambda_1\Theta_1^\ddag\Theta_1+\lambda_2\Theta_2^\ddag\Theta_2,$$
where
$$E_0=\frac{1}{2}(\lambda_1+\lambda_2-a_{11}-a_{22}).$$
\end{theorem}
\begin{proof}
Having in mind that
$$H=-{1\over2}(\ta+\tb)+{1\over2}\left[\begin{matrix}a_1^*\\a_2^*\\a_1\\a_2\end{matrix}\right]^TA\Omega^T
\left[\begin{matrix}a_1^*\\a_2^*\\a_1\\a_2\end{matrix}\right]
$$and
$$AU^T=U^T\diag(\lambda_1,\lambda_2,-\lambda_1,-\lambda_2),$$
it follows that
$$A=U^T\diag(\lambda_1,\lambda_2,-\lambda_1,-\lambda_2)\Omega U\Omega^T.$$ Hence, we
easily obtain
\begin{eqnarray*}&&H=-{1\over2}(\ta+\tb)+{1\over2}\left[\begin{matrix}a_1^*\\a_2^*\\a_1\\a_2\end{matrix}\right]^TU^T
\diag(\lambda_1,\lambda_2,-\lambda_1,-\lambda_2)\Omega^TU\left[\begin{matrix}a_1^*\\a_2^*\\a_1\\a_2\end{matrix}\right]
\\&&
=-{1\over2}(\ta+\tb)+{1\over2}\left[\begin{matrix}\Theta_1^\ddag\\\Theta_2^\ddag\\\Theta_1\\\Theta_2\end{matrix}\right]^T
\diag(\lambda_1,\lambda_2,-\lambda_1,-\lambda_2)\Omega^T\left[\begin{matrix}
\Theta_1^\ddag\\\Theta_2^\ddag\\\Theta_1\\\Theta_2\end{matrix}\right]
\\&&
={1\over2}(\lambda_1+\lambda_2-\ta-\tb)+\lambda_1\Theta_1^\ddag\Theta_1+\lambda_2\Theta_2^\ddag\Theta_2.
\end{eqnarray*}
The lowest eigenvalue of $H$ is
$E_0=(\lambda_1+\lambda_2-\ta-\tb)/2$. \end{proof}

\begin{rmk}In the previous Theorem, we have expressed the
Hamiltonian in terms of the number operators for pseudo-bosons,
$$\widehat N_i=\Theta_i^\ddag\Theta_i,~~i=1,2,$$whose eigenvalues are
the integers $n_i=0,1,2,\ldots.$
\end{rmk}
\begin{corol}
The eigenfunctions of $H$ are
$$\Psi_{n_1,n_2}=\Theta_1^{\ddag n_1}\Theta_2^{\ddag n_2}\Psi_0,~~n_1,n_2\in\Z^+,$$
for $\Theta^\ddag_{i},$  $i=1,2,$ as in (\ref{Thetaddag}),
and $\Psi_0$ satisfies
$$\Theta_i\Psi_0=0,$$
for  $\Theta_{i},$  $i=1,2$, as in (\ref{Theta}).

The corresponding eigenvalues of $H$ are given by
$$E_{n_1,n_2}=E_0+n_1\lambda_1+n_2\lambda_2,$$
where
 \begin{eqnarray*}&&E_0=\frac{1}{2}(\lambda_1+\lambda_2-a_{11}-a_{22}).
\end{eqnarray*}
\end{corol}
\color{black}
\begin{proof}
The result follows from Theorem \ref{T4.1}. 
It is clear that the vacuum $0\neq\Psi_0\in\cal D$ of the operators
$\Theta_1,\Theta_2$, satisfying $\Theta_1\Psi_0=\Theta_2\Psi_0=0$,
is an eigenvector of $H$ with associated eigenvalue $E_0$.
\end{proof}\\

\section{Biorthogonal systems of eigenvectors}\label{S5}
We show that creation  and annihilation operators of $H^*$ are
described by $\Theta_i^*$ and $\Theta_{i}^{\ddag*}$, defined as follows
\begin{eqnarray*}
&&\Theta_{i}^*=v_{i+2,1}a_1+v_{i+2,2}a_2+v_{i+2,3}a^*_1+v_{i+2,4}a^*_2,~~i=1,2,\\
&&\Theta^{\ddag*}_i=v_{i,1}a_1+v_{i,2}a_2+v_{i,3}a^*_1+v_{i,4}a^*_2,~~i=1,2,
\end{eqnarray*}
with $v_i$ as in (\ref{vi}).
Applying the equation of motion method to $H^*$ we are led to
consider the matrix
$$\widetilde A=\left[\begin{matrix}\a_{11}&-\a_{12}&\b_{11}&\b_{12}\\
\a_{12}&\a_{22}&\b_{12}&\b_{22}\\
\b_{11}&\b_{12}&-\a_{11}&-\a_{12}\\
\b_{12}&\b_{22}&\a_{12}&-\a_{22}\end{matrix}\right].$$
It can be easily seen that the
eigenvectors of $\widetilde A$ and $A$ are related as follows.
\begin{pro}If $(\widetilde
x_1,\widetilde x_2,\widetilde y_1,\widetilde y_2)^T$ is an
eigenvector of $\widetilde A$ associated with the eigenvalue
$\lambda$, then $(\widetilde y_1,\widetilde y_2,-\widetilde
x_1,-\widetilde x_2)^T$ is an eigenvector of $A$  associated with
the eigenvalue $-\lambda$.\end{pro} 
There is an interesting inter relation between dynamical creation and
annihilation operators of $H$ and $H^*$.
Suppose $\widetilde \Theta_i^\ddag=\Theta_i^*$ is a creation operator
corresponding to $H^*$. Then, $\widetilde
\Theta_i^{\ddag*}=\Theta_i$ is an annihilation operator
corresponding to $H$. Similarly, if $\widetilde
\Theta_i=\Theta_i^{\ddag*}$ is an annihilation operator
corresponding to $H^*$, then $\widetilde \Theta_i^*=\Theta^\ddag_i$
is a creation operator corresponding to $H$.
Hence, the following result is easily obtained,
\begin{theorem}
The vacuum $0\neq\widetilde\Psi_0\in{\rm span}{\cal F}_{\widetilde\Psi}$ of the operators
$\Theta_1^{\ddag*},\Theta_2^{\ddag*}$, satisfying
$\Theta_1^{\ddag*}\widetilde\Psi_0=\Theta_2^{\ddag*}\widetilde\Psi_0=0$,
is an eigenvector of $H^*$. The associated eigenvalue
coincides with the lowest eigenvalue $E_0=\frac{1}{2}(\a_{11}+\a_{22}-\lambda_1-\lambda_2)$ of $H$. The
eigenvectors of $H^*$ are
$$\widetilde\Psi_{n_1,n_2}=\Theta_1^{* n_1}\Theta_2^{* n_2}\widetilde\Psi_0,~~n_1,n_2\in \Z^+$$
with
associated eigenvalues
$$E_0+n_1\lambda_1+n_2\lambda_2.$$
\end{theorem}
\begin{proof}
Analogous to the Proof of Theorem \ref{T4.1}
\end{proof}\\

We show that the
obtained systems of eigenvectors of $H$ and $H^*$ are complete. Since
$$\langle\Psi_{n_1,n_2},\widetilde\Psi_{n'_1,n'_2}\rangle=
\langle\Theta_1^{\ddag n_1}\Theta_2^{\ddag n_2}\Psi_{0},\Theta_1^{*
n'_1}\Theta_2^{* n'_2}\widetilde\Psi_{0}\rangle=\langle\Theta_1^{
n'_1}\Theta_2^{ n'_2}\Theta_1^{\ddag n_1}\Theta_2^{\ddag
n_2}\Psi_{0},\widetilde\Psi_{0}\rangle,$$it follows that the
biorthogonality relations hold
$$\langle\Psi_{n_1,n_2},\widetilde\Psi_{n'_1,n'_2}\rangle=n_1!n_2!\delta_{n_1n'_1}\delta_{n_1n'_1}
\langle\Psi_{0},\widetilde\Psi_{0}\rangle,\quad
n_1,n_2,n_1',n_2'\geq 0,$$ where $\delta_{ij}$ denotes the
Kroenecker symbol ($\delta_{ij}=1$ for $i=j$, and $0$ otherwise).
This means that the sets of eigenfunctions
\begin{equation}{\cal
F}_\Psi=\{\Psi_{n_1,n_2}:n_1,n_2\geq0\}
\label{*}
\end{equation}and
\begin{equation}{\cal
F}_{\widetilde\Psi}=\{\widetilde\Psi_{n_1,n_2}:n_1,n_2\geq0\}
\label{**}\end{equation}
are complete.
The set ${\cal F}_\Psi$ is complete, since $0$ is the only vector
orthogonal to all its vectors. An analogous observation holds for
${\cal F}_{\tilde\Psi}$. Completeness is equivalent to ${\cal
F}_\Psi$ being a basis  of $\cal H$ if ${\cal F}_\Psi$ is an
orthonormal set in the finite dimensional setting, but not in general, when $\cal H$ is infinite
dimensional.

\section{Physical Hilbert space}\label{S6}
In order to define the {\it physical Hilbert space}, we consider
the subspace ${\cal H}_S$ of $\cal H$ generated by the vector system ${\cal F}_\Psi$ and the
operator $\e^S:{\rm span}~{\cal F}_{\Psi}
 \rightarrow{\rm span}~\widetilde{\cal F}_\Psi$
 defined by
 $$\widetilde\Psi_{mn}=\e^S\Psi_{mn},$$
 which is  positive definite as $\langle\e^S\Psi,\Psi\rangle>0$ for any $0\neq\Psi\in \cal H$.
 Thus
 $$\Psi_{mn}=\e^{-S}\widetilde\Psi_{mn}. $$
\color{black}

We introduce in ${\cal H}_S$ the inner product
$$\langle\Psi,\Phi\rangle_S=\langle\e^{-S}\Psi,\Phi\rangle,~~\Phi,~\Psi\in{\cal H}_S.$$
This inner product
is crucial for the
implementation of the probabilistic interpretation of quantum
mechanics. The {\it physical Rayleigh quotient}
$$\frac{\langle H\Psi,\Psi\rangle_S}{\langle\Psi,\Psi\rangle_S},~~\Psi\in{\cal H}_S,$$
is real and represents the expectation value of the energy restricted to ${\cal H}_S$.

We also define the {\it physical numerical range} of an operator $O$ as
\begin{equation}W_S(O)=\left\{\frac{\langle O\Psi,\Psi\rangle_S}{\langle\Psi,\Psi\rangle_S}:~\Psi\in{\cal H}_S\right\}.\label{WPhys}\end{equation}
With respect to the inner product $\langle\Psi,\Phi\rangle_S,$  the operator $H$ is Hermitian. Indeed, we have
\begin{eqnarray*}
&&\langle\e^{-S}H\Psi_{mn},\Psi_{m'n'}\rangle=(E_0+n\lambda_1+m\lambda_2)\langle\e^{-S}\Psi_{mn},\Psi_{m'n'}\rangle\\
&&=(E_0+n\lambda_1+m\lambda_2)\langle\Psi_{mn},\e^{-S}\Psi_{m'n'}\rangle\\
&&=(E_0+n\lambda_1+m\lambda_2)\langle\Psi_{mn},\widetilde\Psi_{m'n'}\rangle\\
&&=(E_0+n\lambda_1+m\lambda_2)\langle\widetilde\Psi_{mn},\Psi_{m'n'}\rangle\\
&&=\langle H^*\widetilde\Psi_{mn},\Psi_{m'n'}\rangle\\
&&=\langle\e^{-S}\Psi_{mn},H\Psi_{m'n'}\rangle.
\end{eqnarray*}

Since
$\langle H^*\widetilde\Psi_{mn},\Psi_{m'n'}\rangle=\langle H^*\e^{-S}\Psi_{mn},\Psi_{m'n'}\rangle,$
we may also write
$$\e^{-S}H=H^*\e^{-S},$$
that is, $\e^{-S}H$ is Hermitian.

With respect to the the inner product $\langle\Psi,\Phi\rangle_S$ the adjoint of the operator operator $\Theta_i$ is the operator
$\Theta_i^\ddag,~~i=1,2$. Indeed, we have
 \begin{eqnarray*}
&&\langle\e^{-S} \Theta_1\Psi_{nm},\Psi_{n'm'}\rangle=\langle\e^{-S}\Psi_{(n-1),m},\Psi_{n'm'}\rangle\\
&&=\langle\widetilde\Psi_{(n-1),m},\Psi_{n'm'}\rangle=\langle\widetilde\Psi_{n,m},\Psi_{(n'+1),m'}\rangle\\
&&=\langle\e^{-S}\Psi_{nm},\Psi_{(n'+1),m'}\rangle=\langle\e^{-S}\Psi_{n,m},\Theta_1^\ddag\Psi_{n'm'}\rangle.
\end{eqnarray*}
Similarly, we may show that  the adjoint of $\Theta_2$ is
$\Theta_2^\ddag.$ In this framework, the operators $\Theta_i,~\Theta_i^\ddag$ are conventional boson operators.

As required by the {\it Superposition Principle}, any physical
vector $0\neq\Psi\in{\rm span}~{\cal F}_\Psi$ may be uniquely expanded as
$$\Psi=\sum_{n_1,n_2}c_{n_1,n_2}\Psi_{n_1,n_2},$$
where
\begin{eqnarray*}
c_{n_1,n_2}=\frac{\langle\e^{-S}\Psi,\Psi_{n_1,n_2}\rangle}{
\langle\e^{-S}\Psi_{n_1,n_2},\Psi_{n_1,n_2}\rangle}.\end{eqnarray*}
It may be easily verified that
$$\langle\e^{-S}\Psi,\Psi\rangle=\sum_{n_1,n_2}|A_{n_1,n_2}|^2,$$
\color{black}where
$$A_{n_1,n_2}={\langle\e^{-S}\Psi,\Psi_{n_1,n_2}\rangle
\over\sqrt{\langle\e^{-S}\Psi_{n_1,n_2},\Psi_{n_1,n_2}\rangle}}$$
may be interpreted as the probability amplitude. Following Mostafazadeh
\cite{mostafa}, we call the vector space ${\rm span}~{\cal F}_S$ endowed with the
inner product $\langle\e^{-S}\cdot,\cdot\rangle$ the
physical Hilbert space of the model.

\color{black}
\section{Non-Hermitian statistical mechanics}\label{S7}
The main objective of this section is to present the  description, according to
quantum statistical mechanics, of
a system characterized by a non-Hermitian Hamiltonian possessing real eigenvalues
\cite{bebiano**,bebiano***}
.

In quantum statistical mechanics, pure states are represented by vectors and  mixed states are described by  {\it density operators},
i.e., positive semidefinite Hermitian operators with trace 1.
The density operator encapsulates the statistical properties of the system.
Observable quantities are represented by Hermitian operators.
Let us consider a system of $n$ types of bosons characterized by a Hermitian Hamiltonian $H$
and by a conserved {\it number operator}, that is, which commutes with $H$, given by
$$N=\sum_{i=1}^n a_i^* a_i,$$
where the $a_i$ are bosonic operators satisfying the Weil-Heisenberg algebra. At
{\it statistical equilibrium} its density operator is
$$
\rho=\frac{\e^{-\beta H+\zeta N}}{\Tr \e^{-\beta H+\zeta N}},
$$
where $\beta$ is the inverse of the {\it absolute temperature} $T$ and $\zeta$ is related to the so called
{\it chemical potential} $\mu$ according to $\zeta=\beta\mu,$ where $\mu\leq0.$

The {\it partition function} $Z$ is
\begin{equation}Z=\Tr \exp(-\beta H+\zeta N).\label{Z}\end{equation}
According to statistical thermodynamics, the equilibrium properties of the system may be derived
from its logarithm.
It is well known that
$$E=\Tr\rho H=-\frac{\partial\log Z}{\partial\beta}$$
and
$$\langle N\rangle=\Tr N\rho=\frac{\partial\log Z}{\partial Z}$$
are, respectively,
the {\it expected values} of $H$ and  $N$.
The {\it von Neumann entropy} is given by
$S=-\Tr(\rho\log\rho)$. 

In the case of the model we are considering, 
involving two types of bosons, we propose that the role
of the number operator $N$ be played by the pseudo-bosonic number operator
\begin{equation}
\widehat N=\sum_{k=1}^2 d^\ddag_k d_k,\label{Nop}
\end{equation}
which is obviously non-Hermitian, and is in consonance with the corresponding expression for $H$,
\begin{equation}
H=\sum_{k=1}^2\lambda_kd^\ddag_k d_k.\label{Hdddagd}
\end{equation}
In the definition of the partition function, (\ref{Nop}) and (\ref{Hdddagd}) shall be used.
This ensures that $Z$ is real and positive even though $H$ and $\widehat N$ are non-Hermitian.
On the other hand, the density operator
$$
\rho=\frac{\e^{-\beta H+\zeta\widehat N}}{\Tr \e^{-\beta H+\zeta\widehat N}},
$$
is Hermitian with respect to the inner product $\langle\cdot,\cdot\rangle_S.$

The {\it physical numerical range} of $H+i\widehat N$, defined in $\ref{WPhys}$. is $W_S(H+i\widehat N).$

\begin{Exa}
Let us assume that the parameters $\a_{ij},~\beta_{ij}$ which define $H$ are such that  $E_0=1,~\lambda_1=1,~\lambda_2=3$.
The dependence of $\langle\widehat N\rangle_S$ on  $\langle H\rangle_S$ is shown
in Figure 1,  for fixed values of $\beta$, namely,  $\beta=0.125,~0.25,~0.5,~1$ and $4$, and variable $\mu.$
The (open) set of the possible values of the pair  $(\langle H\rangle$, $\langle\widehat N\rangle)$ coincides with the physical numerical range
$W_S(H+i\widehat N)$.
The boundary of the physical numerical range is given by $\{(x,y)=(t,1+t):0\leq t<\infty\}\cup\{(x,y)=(t,1+3t):0\leq t<\infty\}$
and is represented by the dashed lines.
\end{Exa}
\begin{figure}[h]
\centering
\includegraphics[width=0.7\linewidth,angle=0]
{
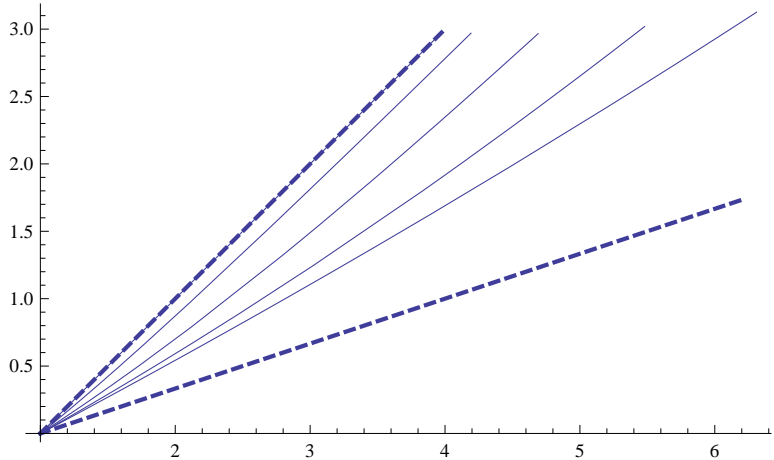}
\caption{Here, $\langle\widehat N\rangle$ vs.  $\langle H\rangle$  is represented.
 It has been assumed that $H$ is such that  $E_0=1,~\lambda_1=1,~\lambda_2=3$. The lines correspond to fixed values of $\beta$, namely, $\beta=0.125,~0.25,~0.5,~1$ and $4$,
from bottom to top, and variable $\mu$.
The (open) set of the possible values of the pair $(\langle H\rangle_S$, $\langle\widehat N\rangle_S)$   coincides with
$W_S(H+i\widehat N)$. 
The boundary of 
$W_S(H+i\widehat N)$ is given by $\{(x,y)=(t,1+t):0\leq t<\infty\}\cup\{(x,y)=(t,1+3t):0\leq t<\infty\}$
and is represented by the dashed lines.
}
\label{fig95z}
\end{figure}
\section{Final remarks}\label{S8}
The existence of the
unbounded operator $\exp S$ such that $\Psi_{n_1n_2}=\exp
S\Phi_{n_1n_2}$, does not ensure that ${\cal F}_\Psi$ constitutes a
basis of $\cal H$.
However, despite this
negative result, ${\cal F}_\Psi$ is a basis of the so called
physical Hilbert space \cite{mostafa}, which is a subset of $\cal
H$, endowed with an adequate inner product.
The introduced physical Hilbert space
allows for the probabilistic interpretation of the model, according to quantum mechanics.

Statistical thermodynamics  considerations, in which the physical Hilbert space plays an important role, are applied to
the studied non-Hermitian Hamiltonian. Expressions for the energy
and number operator expectation values, in terms of the absolute temperature $T$ and of the chemical potential $\mu$,
have been given.

The Hamiltonian $H$ we have considered describes a system of two interacting bosons.
The results here obtained may be extended to the case of $n$ interacting bosons.
\section*{Acknowledgments}
This work was partially supported by the Centre for Mathematics of the
University of Coimbra - UIDB/00324/2020, funded by the Portuguese
Government through FCT/MCTES.

\end{document}